\documentclass[12pt,oneside]{article}
\usepackage{newlfont}
\usepackage{amsmath}
\usepackage{amssymb}
\usepackage{times}
\usepackage{array,multirow,graphicx}
\usepackage{authblk}
\usepackage{amsthm}
\usepackage{ctable}
\usepackage{tabulary}
\usepackage{multirow}
\usepackage{float}
\newtheorem{thm}{Theorem}[section]
\newtheorem{cor}[thm]{Corollary}

\newcommand{\e}[2]{{\mathop{#1}\limits_{#2}}}

\newcommand{\s}[1]{{\mathop{#1}\limits^{*}}}

\newcommand{\p}[1]{{\mathop{#1}\limits_{+}}}
\newcommand{\m}[1]{{\mathop{#1}\limits_{-}}}

\newcommand{\edf}{{\mathop{=}\limits^{\textmd{def.}}}}
\newcommand\overcirc[1]{\raisebox{10pt}{\tiny$\circ$}{\kern-6.5pt}\mbox{$#1$}}
\newcommand\undersym[2]{\raisebox{-6pt}{\tiny$#2$}{\kern-5pt}\mbox{$#1$}}
\numberwithin{equation}{section}
\date{}
\begin{document}
\title{Einstein Geometrization Philosophy and Differential Identities in PAP-Geometry}%
\author[1,4]{M. I. Wanas \thanks{miwanas@sci.cu.edu.eg}}
\author[2,4]{Nabil L. Youssef \thanks{nlyoussef@sci.cu.edu.eg\,;\, nlyoussef2003@yahoo.fr}}
\author[3,4]{W. El Hanafy \thanks{waleed.elhanafy@bue.edu.eg}}
\author[1,4]{S. N. Osman\thanks{samah@sci.cu.edu.eg}}
\affil[1]{\small \it Astronomy department, faculty of science, Cairo university, Egypt}
\affil[2]{\small \it Mathematics department, faculty of science, Cairo university, Egypt}
\affil[3]{\small \it Centre for theoretical physics, the British University in Egypt, 11837 - P.O. Box 43, Egypt}
\affil[4]{\small \it Egyptian Relativity Group, Cairo University, Giza 12613, Egypt}
\renewcommand\Authands{ and }
\maketitle
\begin{abstract}
The importance of Einstein's geometrization philosophy, as an alternative to the least action principle, in constructing general relativity (GR), is illuminated. The role of differential identities in this philosophy is clarified. The use of Bianchi identity to write the field equations of GR is shown. Another similar identity in the absolute parallelism geometry is given. A more general differential identity in the parameterized absolute parallelism geometry is derived. Comparison and interrelationships between the above mentioned identities and their role in constructing field theories are discussed.
\end{abstract}




\section{Introduction}\label{1}
In the second decade of the twentieth century, Einstein constructed a successful theory for gravity, the general theory of relativity (GR) \cite{E1915-2}. Although an action principle has been used afterwards to derive the equations of the theory, Einstein has used the geometrization philosophy in constructing his theory. Among other principles of this philosophy is that laws of nature are just differential identities in the chosen geometry. The geometry used by Einstein at that time was Riemannian geometry, in which second Bianchi identity is a differential one. This identity can be written in its contracted form as
\begin{equation}\label{Bianchi}
G^{\alpha}_{.~\beta ; \alpha}\equiv 0 ,
\end{equation}
where $G_{\alpha \beta}$ is the Einstein tensor and the semicolon denotes the covariant differentiation using Christoffel symbols (The "dot" in the left hand side of \eqref{Bianchi} means that when we lower the upper index,  it occupies the place of the "dot").
Einstein has considered the identity \eqref{Bianchi} as a geometric representation of the law of conservation of matter and energy which is written as
\begin{equation}\label{Conservation}
T^{\alpha}_{.~\beta , \alpha}= 0 ,
\end{equation}
where $T^{\alpha \beta}$ is a symmetric second order tensor describing the material-energy distribution. The comma in \eqref{Conservation} denotes partial differentiation with respect to $x^\alpha$. In order for \eqref{Conservation} to represent a tensor, Einstein has replaced this comma by a semicolon and wrote the field equations of his theory as
\begin{equation*}\label{GR}
  G_{\mu \nu}=-\kappa T_{\mu \nu},
\end{equation*}
where $\kappa$ is a conversion constant. This shows the importance of differential identities in constructing GR.

In the third decade of the twentieth century, Einstein \cite{E28-1, E28-2} used another type of geometry, absolute parallelism (AP-) geometry in his attempt to construct a field theory unifying gravity and electromagnetism. It is characterized by a  non-symmetric linear connection, the  Weitzenb\"{o}ck connection, with vanishing curvature. After long discussions and correspondence between A. Einstein and E. Cartan \cite{ECletters} for more than three years, Einstein stopped using AP-geometry claiming that his attempt is not a successful one. This geometry has been neglected for about two decades until McCrea and Mikhail have reconsidered it to treat continuous creation of matter \cite{M52} and to build a steady state world model \cite{MM56}.

In the sixties of the past century, Dolan and McCrea (1963, private communication to the first author in 1973) have developed a variational method to get fundamental identities in Riemannian geometry, without using an action principle. Unfortunately, this method is not published although it is shown to be of importance when applied in geometries other than the Riemannian one.

In the seventies of the twentieth century two important results have been obtained by developing AP-geometry in a certain direction \cite{MW77, W75}. The first is the use of the dual connection (the transposed Weitzenb\"{o}ck connection) which has simultaneously non-vanishing torsion and non-vanishing curvature. The second result is a differential identity, in AP-geometry, obtained by developing the Dolan-McCrea method to suit AP-space with its dual connection. The identity obtained can be written in the form
\begin{equation}\label{AP-identity}
  E^{\mu}{_{{\nu}|\m\mu}}\equiv 0 \,,
\end{equation}
where the stroke $"|"$ and the ``--'' sign are used to denote covariant differentiation using the dual connection and $E^{\mu}_{.~\nu}$ is a second order non-symmetric tensor defined in the context of AP-geometry by \cite{MW77}
\begin{equation}\label{E-definition}
  \lambda \, E^{\mu}_{.~\nu} \overset{\text{def.}}{=}\frac{\delta \cal{L}}{\delta \underset{i}{\lambda}\, _{\mu}} \underset{i}{\lambda}\, _{\nu}\, ,
\end{equation}
where $\frac{\delta \cal{L}}{\delta \underset{i}{\lambda}\, _{\mu}}$ is the Hamiltonian derivative of the Lagrangian density $\cal{L}$ with respect to the building blocks (BB) $\underset{i}{\lambda}\, _{\mu}$ of AP-geometry  and $\lambda$ is the determinant of $(\underset{i}{\lambda}^\mu)$ (more details about  AP-geometry is given in the next section).

It is to be noted that the identity \eqref{AP-identity} is a generalization of the Bianchi identity \eqref{Bianchi} to AP-geometry. Einstein tensor $G^{\alpha}_{.~\beta}$ is uniquely defined in Riemannian geometry while $E^{\alpha}_{.\beta}$ is not uniquely defined in the AP-geometry. The later depends on the choice of the Lagrangian density $\cal L$ which has many forms in AP-geometry. Several field theories have been constructed \cite{MW77, WS2010, WSR}, using the geometrization philosophy and the identity \eqref{AP-identity}. Many successful applications have been obtained within these theories \cite{SO2010, SO2011, Mazumder-Ray90, MWE95, W85, W2007, WS2013}.

In the last decade of the twentieth century another important result has been obtained as a consequence of another development in AP-geometry. It has been discovered that AP-geometry admits a hidden parameter \cite{WMK95}. It was the Bazanski approach \cite{Bazanski89} that when applied to AP-geometry has led to the above mentioned hidden parameter. The generalization of this hidden parameter has given rise to a modified version of AP-geometry called  parameterized absolute parallelism (PAP-) geometry \cite{W98, W2000}.

In the past ten years, or so, AP-geometry has gained a lot of attention in constructing field equations for what is called $f(T)$ theories and their applications (c.f. \cite{Nashed-KN15, Nashed-torsion-matter15, Nashed-ElHanafy14}). The scalar $T$ is a torsion scalar characterizing the teleparallel equivalent of GR. It has been shown that many other scalars can be defined in AP-geometry and its different versions. These scalars have been used to construct several field theories with GR limits \cite{MW77,WS2010, WKamal2012, WSR, WYElHanafy2014, WYSid-Ahmed2010}.

It is the aim of the present work to find out differential identities in PAP-geometry and to study their relations to the similar identities in AP-geometry \eqref{AP-identity} and in Riemannian geometry \eqref{Bianchi}. For this reason we give a brief review of PAP-geometry in the next section. In section 3 we apply a modified version of the Dolan-McCrea variational method in PAP-geometry to get the general form of the identities admitted. In section 4 we present the interrelationships between the identities in Riemannian geometry, AP-geometry and PAP-geometry. We discuss the results obtained and give some concluding remarks in section 5.
\section{Brief review of PAP-geometry}\label{S2}
In the present section we are going to give basic mathematical machinery and formulae of PAP-geometry necessary for the present work. For more details, the reader is referred to
\cite{W98, W2000} and the references therein.

As mentioned in section 1, it has been shown \cite{WMK95} that the conventional AP-geometry admits a hidden jumping parameter. This parameter has been discovered when using the Bazanski approach \cite{Bazanski89} to derive path and path deviation equations. The importance of this parameter is that, among other things, its value jumps by a step of one-half, which has been tempted to show its relation to some quantum phenomena when used in application. This parameter has no explicit appearance in the conventional AP-geometry. This motivates authors to give the parameter an explicit appearance in a modified version of  AP-geometry known in the literature as PAP-geometry.

An AP-space is a pair  $(M,\, \undersym{\lambda}{i})$, where $M$ is a smooth manifold and\, $\undersym{\lambda}{i}$ are $n$ independent global vector fields $\undersym{\lambda}{i} (i=1,...,n)$ on $M$. For more details, see \cite{M62, W2001, NW13, NA07, NA08}. An AP-space admitted  at least four natural\footnote{"Natural" means that the geometric object under consideration is constructed from the building blocks $\,\undersym{\lambda}{i} $ only.} linear connections: the Weitzenb\"{o}ch connection
\begin{equation}\label{canonical}
\Gamma^{\alpha}_{~\mu\nu}\overset{\text{def.}}{=}
\undersym{\lambda}{i}^{\alpha}\,\undersym{\lambda}{i}_{\mu,\nu},
\end{equation}
its dual (or transposed)
\begin{equation}\label{dual}
\widetilde{{\Gamma}}^{\alpha}_{~\mu\nu}\overset{\text{def.}}{=} \Gamma^{\alpha}_{~\nu\mu},
\end{equation}
the symmetric part of \eqref{canonical}
\begin{equation}\label{symmetric}
\Gamma^{\alpha}_{~(\mu\nu)}\overset{\text{def.}}{=}
\frac{1}{2}(\Gamma^{\alpha}_{~\mu\nu} + \Gamma^{\alpha}_{~\nu\mu})
\end{equation}
and the Levi-Civita connection
\begin{equation}\label{Christoffel}
\{ ^{~\alpha}_{\mu ~\nu}\} \overset{\text{def.}}{=} \frac{1}{2}g^{\alpha \sigma}\left(g_{\mu \sigma, \nu}+g_{\nu \sigma,\mu}-g_{\mu \nu,\sigma}\right)
\end{equation}
associated with the metric  $g_{\mu \nu} \overset{\text{def.}}{=} \e{\lambda}{i}{_{\mu}}\e{\lambda}{i}{_{\nu}}$.

It is to be noted that all these linear connections have non-vanishing curvature except \eqref{canonical}. Among these connections, the only connection that has simultaneously  non-vanishing torsion and non-vanishing curvature  is the dual connection \eqref{dual}. In order to generalize these connections and to give an explicit appearance of the above mentioned jumping parameter, these connections have been linearly combined to give the object ($a_{1}\,, a_{2}\,,a_{3}\,,a_{4}$ are parameters):
\begin{equation}\label{p-gen-conn}
\nabla^\alpha_{~\mu\nu}=a_1 \Gamma^\alpha_{~\mu\nu}+ a_2 \widetilde{\Gamma}^\alpha_{~\mu\nu}+ a_3 \Gamma^\alpha_{~(\mu\nu)}
+a_4~ \{^{~\alpha}_{\mu~\nu}\}.
\end{equation}
Imposing the general coordinate transformation of a linear connection on \eqref{p-gen-conn}, using \eqref{canonical}-\eqref{Christoffel}, the above 4-parameters reduce to one parameter $b$, which leads to
\begin{equation}\label{p-canonical}
    \nabla^{\alpha}_{~\mu \nu}=\{^{~\alpha}_{\mu~\nu}\}+b\, \gamma^{\alpha}_{~\mu \nu},
\end{equation}
where $\gamma^{\alpha}_{~\mu \nu}$ is the contortion of  AP-geometry. The connection \eqref{p-canonical} will be called the parameterized canonical connection. It is non-symmetric with non-vanishing curvature, i.e.,  it is of Riemann-Cartan type. Moreover, it is a metric connection \cite{W98}.

The parameter $b$ is a dimensionless parameter. For some geometric and physical reasons \cite{W2000}, this parameter is suggested to have the form
\begin{equation}\label{b-def}
  b = \frac{N}{2}\alpha \gamma
\end{equation}
where $N$ is a natural number takeing the value 0, 1, 2, ...; $\alpha$ is the fine structure constant and $\gamma$ is a dimensionless parameter to be fixed by experiment or observation for the system under consideration.  The explicit appearance of the parameter $b$ in \eqref{p-canonical} gives rise to the following advantages to PAP-geometry:
\begin{enumerate}
	 \item In the case of $b=1$, \eqref{p-canonical} reduces to the Weitenb\"{o}ch connection \eqref{canonical} of the conventional AP-geometry.
  \item In the case of $b=0$, the connection \eqref{p-canonical} reduces to the Levi-Civita connection (\ref{Christoffel}).
  \item Between the limits $b=0$ and $b=1$ there is a discrete spectrum of spaces (due to the presence of $\frac{N}{2}$ in \eqref{b-def}), each of which has simultaneously non-vanishing curvature and torsion.
\end{enumerate}

In what follows, we will decorate the tensors containing the parameter $b$, by a star. Also, these tensors will retain the same name, as in the conventional AP-geometry, preceded by the adjective ``parameterized''. For example, the parameterized contortion is given by
\begin{equation}\label{p-contorsion}
    \s{\gamma}{^{\alpha}}_{\mu \nu}= b\, \gamma^{\alpha}_{~\mu \nu}.
    \end{equation}
The torsion of the parameterized connection \eqref{p-canonical} is  given by
\begin{equation}\label{p-torsion}
\s{\Lambda}{^{\alpha}}_{\mu \nu}= b(\gamma^{\alpha}_{~\mu \nu}-\gamma^{\alpha}_{~\nu \mu})=b\,\Lambda{^{\alpha}}_{\mu \nu},
\end{equation}
where $\Lambda^{\alpha}_{~\mu \nu}$ is the torsion of the AP-geometry. Also, the parameterized basic form can be obtained by contraction of \eqref{p-contorsion} or \eqref{p-torsion}
\begin{equation*}\label{p-vector}
\s{C}_{\mu}\edf \s{\Lambda}{^{\alpha}}_{\mu \alpha}=\s{\gamma}^{\alpha}_{~\mu \alpha}=b\,C_{\mu}.
\end{equation*}
It is to be noted that starred objects have the same symmetry properties as the corresponding unstarred objects defined in the conventional AP-geometry.

Now, since the parameterized canonical connection $\nabla^{\alpha}_{~\mu\nu}$ is non-symmetric, one can define from which two other linear connections: the parameterized dual (transposed)  connection
\begin{equation}\label{p-dual}
\widetilde{\nabla}^\alpha_{~\mu\nu}  \overset{\text{def.}}{=} \nabla^\alpha_{~\nu\mu}
\end{equation}
and the  parameterized symmetric connection
\begin{equation}\label{p-symm}
  \nabla^\alpha_{~(\mu\nu)}  \overset{\text{def.}}{=} \frac{1}{2}\,(\nabla^\alpha_{~\mu\nu}+ \nabla^\alpha_{~\nu\mu}).
\end{equation}
Consequently, using \eqref{p-canonical}, \eqref{p-dual} and \eqref{p-symm} we can define respectively the covariant derivatives:
\begin{equation}\label{d-p-canonical}
  A_{\mu\|\p\nu}\edf A_{\mu,\nu} -A_\alpha \nabla^{\alpha}_{.~\mu\nu},
\end{equation}
\begin{equation*}\label{d-p-duall}
  A_{\mu\|\m\nu}\edf A_{\mu,\nu} -A_\alpha \tilde{\nabla}^{\alpha}_{.~\mu\nu},
\end{equation*}
\begin{equation*}\label{d-p-symm}
  A_{\mu\|\nu}\edf A_{\mu,\nu} -A_\alpha \nabla^{\alpha}_{.~(\mu\nu)},
\end{equation*}
for any arbitrary vector $A_{\mu}$.\\

It is of importance to note that the BB of  PAP-geometry are the same as those of  AP-geometry. In other words, the parameter $b$ has no effect on the BB of the geometry.

\section{Dolan-McCrea Variational Method in PAP-geometry}\label{S3}
 We are going to apply the Dolan-McCrea variational method to obtain differential identities in  PAP-geometry. This method is an alternative to the least action method. It is originally suggested in 1963 to derive Bianchi and other identities in the context of Riemannian geometry. It has been generalized \cite{MW77} to AP-geometry and used to derive the identity \eqref{AP-identity}. Here, we generalize this method to PAP-geometry, giving some details since it is not widely known.

Since, as stated above, the BB of PAP-geometry are the same as those of AP-geometry, let us define the Lagrangian function
\begin{equation*}\label{Lsc}
\s{L} = \s{L}(\e{\lambda}{i}{_{\mu}},\e{\lambda}{i}{_{\mu,\nu}, \e{\lambda}{i}{_{\mu,\nu \sigma}}}),
\end{equation*}
constructed from the BB of PAP, $\e{\lambda}{i}{_{\mu}}$, and its first and second derivatives. Consequently, we can write the scalar density as follows
\begin{equation}\label{lag_dens_0}
    \mathcal{\s{L}}_{0} = \mathcal{\s{L}}_{0}(\e{\lambda}{i}{_{\mu}},\e{\lambda}{i}{_{\mu,\nu}}, \e{\lambda}{i}{_{\mu,\nu \sigma}}) = \lambda \s{L},
\end{equation}
where $\mathcal{\s{L}}_{0}$ is a scalar lagrangian density and $\lambda :=\det (\e{\lambda}{i}{^{\mu}})\neq0$.
Now, let us assume that the following integral, defined over some arbitrary $n$-dimensional domain $\Omega$,
\begin{equation*}\label{act_0}
    \s{I}{_{0}} \edf \int_{\Omega} \mathcal{\s{L}}_{0}(x) d^{n}x,
\end{equation*}
is invariant under the infinitesimal transformation
\begin{equation}\label{ltrans}
    \e{\lambda}{i}(x) \rightarrow \e{\lambda}{i}(x) + \epsilon \e{h}{i}(x),
\end{equation}
where $\epsilon$ is an infinitesimal parameter and $\e{h}{i}$ are $n$ 1-forms defined on $M$. Let $\Sigma$ be the $(n-1)$-space such that the n-space region $\Omega$ is enclosed within $\Sigma$. Thus
\begin{equation*}\label{act_eta}
   \s{I}{_{\epsilon}} \edf \int_{\Omega}\mathcal{\s{L}}_{\epsilon} d^{n}x,
\end{equation*}
where
\[d^{n}x = dx^{1} \,dx^{2}...\,dx^{n}
\]
and
\begin{equation}\label{lag_dens_eta}
    \mathcal{\s{L}}_{\epsilon} = \mathcal{\s{L}}\left(\e{\lambda}{i}{_{\mu}}+\epsilon\e{h}{i}{_{\mu}},\e{\lambda}{i}{_{\mu,\nu}}
    +\epsilon\e{h}{i}{_{\mu,\nu}},\e{\lambda}{i}{_{\mu,\nu \sigma}}+\epsilon\e{h}{i}{_{\mu,\nu \sigma}}\right).
\end{equation}
Assuming that $\e{h}{i}{_\mu}$ and $\e{h}{i}{_{\mu,\nu}}$ vanish at all points of $\Sigma$, then
\begin{equation}\label{Ieta-I0}
    \s{I}{_{\epsilon}}-\s{I}{_{0}} = \int_{\Omega}(\mathcal{\s{L}}_{\epsilon}-\mathcal{\s{L}}_{0}) d^{n}x,
\end{equation}
where $\s{I}_{0}$ is the value of $\s{I}_{\epsilon}$ at $\epsilon=0$. Now $\mathcal{\s{L}}_{\epsilon}$ can be written, using Taylor expansion, in the form,
\begin{equation}\label{lag_dens_eta2}
    \mathcal{\s{L}}_{\epsilon}=\mathcal{\s{L}}_{0}+\frac{\partial \mathcal{\s{L}}_{0}}{\partial \e{\lambda}{i}{_{\mu}}}\epsilon \e{h}{i}{_{\mu}}
    +\frac{\partial \mathcal{\s{L}}_{0}}{\partial \e{\lambda}{i}{_{\mu,\nu}}}\epsilon \e{h}{i}{_{\mu,\nu}}
    +\frac{\partial \mathcal{\s{L}}_{0}}{\partial \e{\lambda}{i}{_{\mu,\nu \sigma}}}\epsilon \e{h}{i}{_{\mu, \nu \sigma}}+O(\epsilon^{2}).
\end{equation}
Using (\ref{Ieta-I0})  and (\ref{lag_dens_eta2}) and integrating by parts, it can be shown, after some manipulation, that \eqref{Ieta-I0} can be written as
\begin{equation}\label{q40}
    \s{I}_{\epsilon}-\s{I}_{0} = \epsilon \int_{\Omega} \lambda \frac{\delta \s{L}}{\delta \e{\lambda}{i}{_{\mu}}} \e{h}{i}{_{\mu}} d^{n}x+O(\epsilon^{2}),
\end{equation}
where $\frac{\delta \s{L}}{\delta \e{\lambda}{i}{_{\mu}}}$ is the Hamiltonian derivative defined by
\begin{equation*}\label{Ham_derv}
    \frac{\delta \s{L}}{\delta \e{\lambda}{i}{_{\beta}}} \edf \frac{1}{\lambda}\left[\frac{\partial \mathcal{\s{L}}_{0}}{\partial {\mathop{\lambda}\limits_{i}}{_\beta}}- \frac{\partial}{\partial x^{\gamma}}\left(\frac{\partial \mathcal{\s{L}}_{0}}{\partial {\mathop{\lambda}\limits_{i}}{_{\beta, \gamma}}}\right)+\frac{\partial^2}{\partial x^{\gamma} \partial x^{\sigma}}\left(\frac{\partial \mathcal{\s{L}}_{0}}{\partial {\mathop{\lambda}\limits_{i}}{_{\beta, \gamma \sigma}}}\right)\right].
\end{equation*}
Consider a new set of coordinates $\bar{x}$ related to the $x$-coordinate system by the transformation
\begin{equation}\label{coo_tran}
    \bar{x}^{\mu}=x^{\mu}- \epsilon z^{\mu}(x),
\end{equation}
where $z^\mu$ is a vector field vanishing at all points of $\Sigma$ and $\epsilon$ is independent of $x$. We can write,
\begin{equation*}\label{Jacob1}
    \frac{\partial \bar{x}^\mu}{\partial \bar{x}^\nu}=\delta^{\mu}_{\nu}=\frac{\partial x^\mu}{\partial \bar{x}^\nu}-\epsilon \frac{\partial z^\mu}{\partial x^\sigma} \frac{\partial x^\sigma}{\partial \bar{x}^\nu},
\end{equation*}
\begin{equation}\label{Jacob2}
    \frac{\partial x^\mu}{\partial \bar{x}^\nu}=\delta^{\mu}_{\nu}+\epsilon z^{\mu}{_{,\sigma}}\frac{\partial x^\sigma}{\partial \bar{x}^\nu}.
\end{equation}
Since $z^\mu$ is a vector, its transformation law is
\begin{equation*}
    z^{\mu}(x)=\frac{\partial x^\mu}{\partial \bar{x}^\nu}\bar{z}^{\nu}(\bar{x}).
\end{equation*}
Substituting from (\ref{Jacob2}) into the above equation, we get
\begin{equation*}
    z^{\mu}(x)=\delta^{\mu}_{\nu}\bar{z}^{\nu}(\bar{x})+\epsilon \bar{z}^{\nu}(\bar{x})z^{\mu}{_{,\sigma}}\frac{\partial x^\sigma}{\partial \bar{x}^\nu}.
\end{equation*}
Substituting from (\ref{Jacob2}) again into the above equation gives
\begin{equation*}
    z^{\mu}(x)=\bar{z}^{\mu}(\bar{x})+\epsilon \bar{z}^{\nu}(\bar{x})z^{\mu}{_{,\sigma}}\left(\delta^{\sigma}_{\nu}+\epsilon z^{\sigma}{_{,\alpha}}\frac{\partial x^\alpha}{\partial \bar{x}^\nu}\right).
\end{equation*}
Finally, this equation can be written as
\begin{equation}\label{z_vec}
    z^{\mu}(x)=\bar{z}^{\mu}(\bar{x})+\epsilon \bar{z}^{\nu}(\bar{x})z^{\mu}{_{,\nu}}+O(\epsilon^2).
\end{equation}
Substituting from (\ref{z_vec}) into (\ref{coo_tran}) we get
\begin{equation*}\label{coo_tran2}
    x^\mu=\bar{x}^\mu+\epsilon \bar{z}^{\mu}(\bar{x})+O(\epsilon^2)
\end{equation*}
Consequently, we can write
\begin{equation*}\begin{split}\label{Jacob3}
    \frac{\partial \bar{x}^\mu}{\partial x^\nu}=&\delta^{\mu}_{\nu}-\epsilon z^{\mu}{_{,\nu}},\\
    \frac{\partial x^\mu}{\partial \bar{x}^\nu}=&\delta^{\mu}_{\nu}+\epsilon \bar{z}^{\mu}{_{,\nu}}+O(\epsilon^2).
\end{split}\end{equation*}
Then the Jacobian expressions of the above matrices can be given as
\begin{equation}\begin{split}\label{Jacob4}
    \frac{\partial \bar{x}}{\partial x}=&1-\epsilon z^{\alpha}{_{,\alpha}}+O(\epsilon^2),\\
    \frac{\partial x}{\partial \bar{x}}=&1+\epsilon z^{\alpha}{_{,\alpha}}+O(\epsilon^2),
\end{split}\end{equation}
since the difference between $z^{\mu}(x)$ and ${\bar{z}}^{\mu}(\bar{x})$ would be of order $\epsilon^2$ as shown by \eqref{z_vec}.

Now, we express $\e{\bar{\lambda}}{i}{_{\mu}}(\bar{x})$ in terms of $\e{\lambda}{i}{_{\mu}}(x)$ in two different ways:

  \emph{i} ) Using the vector transformation law:
  \begin{equation}\begin{split}\label{method1}
    \e{\bar{\lambda}}{i}{_{\mu}}(\bar{x})=&\frac{\partial x^{\nu}}{\partial \bar{x}^{\mu}}\e{\lambda}{i}{_{\nu}}(x)\\
    =&\delta^{\nu}_{\mu}\e{\lambda}{i}{_{\nu}}(x)+\epsilon \e{\lambda}{i}{_\nu}(x)\bar{z}^{\nu}{_{,\mu}}+O(\epsilon^2)\\
    =&\e{\lambda}{i}{_{\mu}}(x)+\epsilon \e{\lambda}{i}{_{\nu}}z^{\nu}{_{,\mu}}+O(\epsilon^2).
\end{split}\end{equation}

  \emph{ii} ) Using Taylor expansion:
  \begin{equation}\begin{split}\label{method2}
    \e{\bar{\lambda}}{i}{_{\mu}}(\bar{x})=&\e{\bar{\lambda}}{i}{_{\mu}}(x-\epsilon z)\\
    =&\e{\bar{\lambda}}{i}{_{\mu}}(x)
    -\epsilon \e{\bar{\lambda}}{i}{_{\mu,\nu}}(x) z^{\nu}+O(\epsilon^2)\\
    =&\e{\bar{\lambda}}{i}{_{\mu}}(x)-\epsilon \e{\lambda}{i}{_{\mu,\nu}}(x) z^{\nu}+O(\epsilon^2).
\end{split}\end{equation}
We have dropped the bar in the second term on the right hand side of this equation as the difference would be of order $\epsilon^2$.
Comparing (\ref{method1}) and (\ref{method2}) and using (\ref{ltrans}), we get
\begin{equation}\label{lbar-l}
    \e{\bar{\lambda}}{i}{_{\mu}}(x)-\e{\lambda}{i}{_{\mu}}(x)=\epsilon \e{h}{i}{_{\mu}}+O(\epsilon^2),
\end{equation}
where
\begin{equation}\label{himu}
    \e{h}{i}{_{\mu}}=\e{\lambda}{i}{_{\nu}}z^{\nu}{_{,\mu}}+\e{\lambda}{i}{_{\mu,\nu}}z^{\nu}
    =\e{\lambda}{i}{_{\nu}}z^{{\nu}}{_{\|\m \mu}}+\e{\lambda}{i}{_{{\mu}\|\p\nu}}z^{\nu},
\end{equation}
and the double stroke operator $"\|"$ is defined by \eqref{d-p-canonical}.

 Now, we treat the Lagrangian density in the two different ways mentioned above

  a) Using Taylor expansion:
  \begin{equation}\label{Lbar}
    \overline{\mathcal{\s{L}}}(\bar{x})=\overline{\mathcal{\s{L}}}(x-\epsilon z)
        = \overline{\mathcal{\s{L}}}(x)-\epsilon \mathcal{\s{L}}(x){_{,\alpha}} z^{\alpha} +O(\epsilon^2),
\end{equation}
also, we have dropped the bar in the second term on the right hand side of this equation as the difference would be of order $\epsilon^2$.\\
The Lagrangian density $\overline{\mathcal{\s{L}}}(x)$ in (\ref{Lbar}) means
\begin{equation}\label{Lbar2}
    \overline{\mathcal{\s{L}}}(x) \equiv \overline{\mathcal{\s{L}}}(\e{\bar{\lambda}}{i}{_{\mu}}(x),\e{\bar{\lambda}}{i}{_{\mu,\nu}}(x), \e{\bar{\lambda}}{i}{_{\mu,\nu \sigma}}(x)).
\end{equation}

  b) Using scalar transformation:
\begin{equation}\label{Lsc_trans}
\begin{array}{rl}
 \overline{\s{L}}(\bar{x})&=\s{L}(x),\\
\bar{\lambda} \overline{\s{L}}(\bar{x})&=\bar{\lambda} \s{L}(x) =\lambda \s{L}(x)\frac{\partial x}{\partial \bar{x}},
\end{array}
\end{equation}
then, using (\ref{lag_dens_0}) and (\ref{Jacob4}), we get
\begin{equation}\label{bar-Lx}
    \overline{\mathcal{\s{L}}}(\bar{x})=\mathcal{\s{L}}(x)+\epsilon \mathcal{\s{L}}(x) z^{\alpha}{_{,\alpha}}+O(\epsilon^2).
\end{equation}
Comparing (\ref{Lbar}) and (\ref{bar-Lx}), we have
\begin{equation*}
    \overline{\mathcal{\s{L}}}(x)-\mathcal{\s{L}}(x)=\epsilon(\mathcal{\s{L}} (x) z^{\alpha})_{,\alpha}+O(\epsilon^2).
\end{equation*}

By (\ref{Lsc_trans}), we get
\begin{equation*}
    \int_{\Omega}\overline{\s{L}}d^{n}\bar{x}-\int_{\Omega} \s{L} d^{n}x=0,
\end{equation*}
\begin{equation}\label{int_Lbar-L}
    \int_{\Omega}(\overline{\mathcal{\s{L}}}(x)-\mathcal{\s{L}}(x))d^{n}x=\int_{\Omega}[\epsilon(\mathcal{\s{L}}(x) z^{\alpha})_{,\alpha}+O(\epsilon^2)]d^{n}x.
\end{equation}
Applying Gauss's theorem to convert the volume integral on $\Omega$ ($n$-space) to a surface integral over $\Sigma$ ($(n-1)$-space), the integral vanishes as $z^{\alpha}$ vanishes at all points of $\Sigma$. Now, we can write
\begin{equation*}\label{van_int}
    \begin{split}
      \int_{\Omega}(\mathcal{\s{L}} (x) z^{\alpha})_{,\alpha} d^{n}x=& \int_{\Sigma} \mathcal{\s{L}} (x)z^{\alpha}n_\alpha d^{n-1}x \\
        =& \oint_{\Sigma} \s{L}(x)  z^{\alpha}n_{\alpha} d\Sigma \equiv 0,
    \end{split}
\end{equation*}
where $n_{\alpha}$ is a unit vector normal to $\Sigma$.\\

So, equation (\ref{int_Lbar-L}) shows that
\begin{equation}\label{int-Lbar-L2}
    \int_{\Omega}\left[\overline{\mathcal{\s{L}}}({x}) - {\mathcal{\s{L}}}(x)\right] d^n x=O(\epsilon^2).
\end{equation}
Substituting from (\ref{lbar-l}) into (\ref{Lbar2}), we get
\begin{equation*}\label{Lbar_dens}
    \overline{\mathcal{\s{L}}} (x) = \overline{\mathcal{\s{L}}} \left(\e{\lambda}{i}{_\mu}(x)+\epsilon \e{h}{i}{_\mu}+O(\epsilon^2),\e{\lambda}{i}{_{\mu,\nu}}(x)+\epsilon \e{h}{i}{_{\mu,\nu}}+O(\epsilon^2),\e{\lambda}{i}{_{\mu,\nu \sigma}}(x)+\epsilon \e{h}{i}{_{\mu,\nu \sigma}}+O(\epsilon^2)\right).
\end{equation*}
Using (\ref{lag_dens_eta}),
\begin{equation*}\label{Lbar_dens_Lden_eta}
    \overline{\mathcal{\s{L}}} (x)=\mathcal{\s{L}}_{\epsilon}(x)+O(\epsilon^2),
\end{equation*}
then, we can write \eqref{int-Lbar-L2} as
\begin{equation}\label{Final_int}
    \int_{\Omega}[\mathcal{\s{L}}_{\epsilon}(x)-\mathcal{\s{L}}_{0}(x)]d^n x = O(\epsilon^2).
\end{equation}
Comparing (\ref{Ieta-I0}), \eqref{q40} and (\ref{Final_int}) implies that
\begin{equation}\label{van_int2}
    \int_{\Omega}\lambda\frac{\delta \s{L}}{\delta \e{\lambda}{i}{_{\mu}}}\e{h}{i}{_{\mu}}d^n x \equiv 0.
\end{equation}
Substituting from (\ref{himu}) into (\ref{van_int2}), we write
\begin{equation}\label{Int_id}
    \int_{\Omega}\lambda \left[\frac{\delta \s{L}}{\delta\e{\lambda}{i}{_{\mu}}}~\e{\lambda}{i}{_{\nu}}~z^{{\nu}}{_{\|\m \mu}}+\frac{\delta \s{L}}{\delta\e{\lambda}{i}{_{\mu}}}~\e{\lambda}{i}{_{{\mu}\|\p \nu}}~z^{\nu}\right]d^n x \equiv 0.
\end{equation}
Let
\begin{equation}\label{HJ}
    \begin{split}
    \s{E}{^\mu}{_\nu} & \edf \frac{\delta \s{L}}{\delta \e{\lambda}{i}{_\mu}}\e{\lambda}{i}{_\nu}= - \e{\s{J}}{i}{^\mu} \e{\lambda}{i}{_\nu},\\
    \e{\s{J}}{i}{^\mu}& \edf -\frac{\delta \s{L}}{\delta \e{\lambda}{i}{_\mu}}.
    \end{split}
\end{equation}
The integral identity (\ref{Int_id}) then becomes,
\begin{equation}\label{Int_id2}
    \int_{\Omega}\lambda \left[\s{E}{^{\mu}}{_{\nu}}z^{{\nu}}{_{\|\m \mu}}-\e{\s{J}}{i}{^{\mu}}\e{\lambda}{i}{_{{\mu}\| \p\nu}}z^{\nu}\right]d^n x \equiv 0.
\end{equation}

Now, the integral of the first term can be treated as follows
\begin{equation*}\begin{split}
    \int_{\Omega} \lambda \s{E}{^{\mu}}{_{\nu}} z^{\m{\nu}}{_{\| \mu}} d^n x&=\int_{\Omega} \lambda (\s{E}{^{\mu}}{_{\nu}} z^{\nu}){_{\| \m{\mu}}} d^n x - \int_{\Omega} \lambda \s{E}{^{\mu}}{_{\nu \| \m{\mu}}} z^{\nu} d^n x,\\
    &=\int_{\Omega} (\lambda \s{E}{^{\mu}}{_{\nu}} z^{\nu}){_{,\mu}} d^n x - \int_{\Omega} \lambda \s{E}{^{\mu}}{_{\nu \| \m{\mu}}} z^{\nu} d^n x,\\
    &=\oint_{\Sigma} \lambda \s{E}{^{\mu}}{_{\nu}} z^{\nu} n_{\mu} d \Sigma
- \int_{\Omega} \lambda \s{E}{^{\mu}}{_{\nu \| \m{\mu}}} z^{\nu} d^n x,\\
&=- \int_{\Omega} \lambda \s{E}{^{\mu}}{_{\nu \| \m{\mu}}} z^{\nu} d^n x.
\end{split}\end{equation*}
Thus the integral identity (\ref{Int_id2}) reads
\begin{equation*}\label{Int_id3}
    \int_{\Omega} \lambda (\s{E}{^{\mu}}{_{\nu \| \m{\mu}}}+\e{\lambda}{i}{_{{\mu}\| \p\nu}}\e{\s{J}}{i}{^\mu}) z^{\nu} d^n x\equiv 0,
\end{equation*}
since $\lambda \neq 0$ and $z^{\nu}$ is an arbitrary vector field. Using the fundamental lemma of the calculus of variation \cite{CoV}, the expression in the brackets vanishes identically. Hence, we have
\begin{equation}\label{Identity}
    \s{E} {^{\mu}}{_{\nu \| \m{\mu}}}+\e{\lambda}{i}{_{{\mu}\| \p\nu}}\e{\s{J}}{i}{^\mu} \equiv 0.
\end{equation}
Using \eqref{p-canonical}, \eqref{d-p-canonical} and \eqref{HJ}, after some manipulations, the identity \eqref{Identity} can be written as
\begin{equation}\label{Identity3}
    \s{E}{^{\mu}}{_{\nu \| \m{\mu}}} = b(1-b) E^{\mu\alpha}~\gamma_{\alpha\mu\nu}.
\end{equation}
This is the differential identity characterizing PAP-geometry.
\section{Relation between Differential Identities}\label{S4}
The three differential identities, \eqref{Bianchi} of Riemannian geometry, \eqref{AP-identity} of AP-geometry and \eqref{Identity3} of PAP-geometry are to be related. This idea comes from the fact that the general linear connection \eqref{p-canonical} reduces to the Weitzenb\"ock connection \eqref{canonical} of the AP-space  for $b=1$ and to the Live-Civita connection \eqref{Christoffel} of the Riemannian space for $b=0$. Note that the three identities mentioned are derivatives in which linear connections are used. The following two results relate these identities.

\begin{thm}
If $\s{E}{^{\mu\nu}}$ is a tensor of order 2 on a PAP-space $(M,\,\undersym{\lambda}{i}\,)$, then we have the identity
\begin{equation}\label{prop1}
\s{E} {^{\mu}}{_{\nu \|\m{\mu}}}= \s{E} {^\mu}_{\nu;\mu}+b\,\s{E} {^{\mu\alpha}}\,\gamma_{\mu\alpha\nu}.
\end{equation}
Consequently, if $b=0$ (the Riemannian case) or $\s{E} {^{\mu\nu}}$ is symmetric, we get $$\s{E}{^{\mu}}{_{\nu \|\m{\mu}}}= \s{E} {^\mu}_{\nu;\mu}$$
\end{thm}
\begin{proof} Using the definition of the parameterized dual connection $\tilde{\nabla}^{\alpha}_{\,\,\,\,\mu\nu}$, we have
\begin{eqnarray*}
  \s{E}{^{\mu}}{_{\nu \|\m{\mu}}} &=& \s{E} {^\mu}_{\nu,\mu} - \s{E} {^\mu}_{\alpha} \nabla^{\alpha}_{\mu\nu} +
  \s{E} {^\alpha}_{\nu} \nabla^{\mu}_{\mu\alpha} \\
    &=& \s{E} {^\mu}_{\nu,\mu} - \s{E} {^\mu}_{\alpha}\,\{ ^{~\alpha}_{\mu ~\nu}\} - b\,\s{E} {^\mu}_{\alpha}\gamma^{\alpha}_{\mu\nu}
    + \s{E} {^\alpha}_{\nu}\, \{ ^{~\mu}_{\mu ~\alpha}\}+ b\,\s{E} {^\alpha}_{\nu}\gamma^{\mu}_{\mu\alpha}
\end{eqnarray*}
As $\gamma_{\mu\nu\sigma}$ is skew-symmetric in the first pair of indices, the last term on the right vanishes. Hence, we get
$$\s{E}{^{\mu}}{_{\nu \|\m{\mu}}}= \s{E} {^\mu}_{\nu;\mu}- b\,\s{E} {^\mu}_{\alpha}\gamma^{\alpha}_{\,\,\,\mu\nu}.$$
Now, consider the last term on the right of the above identity:\\ $\s{E} {^\mu}_{\alpha}\gamma^{\alpha}_{\,\,\,\mu\nu}
= \s{E} {^\mu}_{\alpha}\delta^{\alpha}_{\epsilon}\gamma^{\epsilon}_{\,\,\,\mu\nu}
= \s{E} {^\mu}_{\alpha}\,g^{\alpha\beta}\,g_{\epsilon\beta}\,\gamma^{\epsilon}_{\,\,\,\mu\nu}
= \s{E} {^{\mu\beta}}\gamma_{\beta\mu\nu}=-\s{E} {^{\mu\beta}}\gamma_{\mu\beta\nu}$,\\
from which \eqref{prop1} follows.\\
Finally, it is clear that if $b=0$, then $\s{E}{^{\mu}}{_{\nu \|\m{\mu}}}=E{^{\mu}}{_{\nu ;{\mu}}}$. On the other hand, if $\s{E}{^{\mu\nu}}$
is symmetric, then the term $\s{E} {^{\mu\alpha}}\gamma_{\mu\alpha\nu}$ vanishes since $\gamma_{\mu\alpha\nu}$ is skew-symmetric in $\mu$ and $\alpha$.
\end{proof}

\begin{thm}
If $\s{E} {^{\mu\nu}}$ is a tensor of order 2 on a PAP-space $(M,\,\undersym{\lambda}{i}\,)$, then we have the identity
\begin{equation}\label{prop2}
\s{E} {^{\mu}}{_{\nu \|\m{\mu}}}= \s{E} {^\mu}{_{\nu |\m{\mu}}}+(b-1)\,\s{E} {^{\mu\alpha}}\,\gamma_{\mu\alpha\nu}.
\end{equation}
Consequently, if $b=1$ (the AP case) or $\s{E}{^{\mu\nu}}$ is symmetric, we get \begin{equation*}
\s{E}{^{\mu}}{_{\nu\|\m{\mu}}}=\s{E} {^\mu}{_{\nu |\m{\mu}}}
\end{equation*}
\end{thm}
The above two theorems imply the following result.
\begin{cor}~\\
\textbf{(a)} For a symmetric tensor $E^{\mu\nu}$, we have from \eqref{prop1} and \eqref{prop2}
\begin{equation}\label{corolla3}
\s{E}{^{\mu}{_{\nu\|\m{\mu}}}}= \s{E} {^\mu}{_{\nu |\m{\mu}} = \s{E} {^\mu}_{\nu;\mu}}.
\end{equation}
\textbf{(b)} In the Riemannian case, the tensor $E^\mu{_\nu}$ defined by \eqref{E-definition} (using the Ricci scalar in the Lagrangian function) coincides with the Einstein tensor, then
\begin{equation}\label{4.4}
E^{\mu}{_ {\nu ; \mu}}\equiv 0,
\end{equation}
and, consequently, for any symmetric tensor defined by \eqref{E-definition}, we have
\begin{equation}\label{4.5}
E^{\mu}{_ {\nu | \m\mu}}\equiv 0\, , \quad \s{E}{^{\mu}}{_ {\nu \| \m\mu}}\equiv 0
\end{equation}
in AP and PAP-spaces, respectively.
\end{cor}
\section{Discussion and Concluding Remarks}

In the present work, we have used the Dolan-McCrea variational method to derive possible differential identities in PAP-geometry. The importance of this work for physical applications can be discussed in the following points:

\begin{itemize}
  \item [\textbf{(1)}] It is well known that the field equations of GR can be obtained using either one of the following approaches:

   \textbf{a)} \emph{The Einstein approach}: in which the geometrization philosophy plays the main role in constructing the field equations of the theory.
   One of the principles of this philosophy is that "\emph{Laws of nature are just differential} \emph{identities in an appropriate geometry}".
   Einstein has used this principle to write his field equations. In other words, he has used the second Bianchi identity of Riemannian geometry to
   construct the field equations of his theory.

  \textbf{ b)} \emph{Hilbert approach}: in which the standard method of theoretical physics, the action principle, has been used to derive the field equations
   of the theory.

   It is to be noted that the first approach is capable of constructing a complete theory, not only the field equations of the theory. This will be discussed in the following point.

  \item [\textbf{(2)}] Einstein geometrization philosophy can be summarized as follows \cite{W2007-acc.}:

``To understand \underline{Nature}, one has to start with \underline{Geometry} and end with physics''

In applying this philosophy one has to consider its main principles:

\textbf{a)} There is a one-to-one correspondence between gometric objects and physical quantities.

\textbf{b)} Curves (paths) in the chosen geometry are trajectories of test particles.

\textbf{c)} Differential identities represent laws of nature.

In view of the above principles, Einstein has used Riemannian geometry to construct a full theory for gravity, GR, in which the
metric and the curvature tensors represent the gravitational potential and strength, respectively. He has also used the geodesic equations to represent motion
in gravitational fields. Finally, he has used Bianchi identity to write the field equations of GR.

The above mentioned philosophy can be applied to any geometric structure other than the Riemannian one. It has been applied successfully in the context of
conventional AP-geometry (cf. \cite{MW77, WS2010, WSR}), using the differential identity \eqref{AP-identity}.

PAP-geometry is more general than both Riemannian and conventional AP-geometry. It is shown in section 2  that PAP-geometry reduces to AP-geometry
for $b=1$ and to Riemannian geometry for $b=0$. In other words, the later two geometries are special cases of the PAP-one. The importance of the parameter
$b$ has been investigated in many papers (cf. \cite{SM04, W2012}). The value of this parameter is extracted from the results of three different
experiments \cite{MTGuthC2007, SM04, WMK2000}. The value obtained ($b=10^{-3}$) shows that space-time near the Earth is neither Riemannian ($b=0$)
nor conventional AP ($b=1$).

  \item [\textbf{(3)}] To use PAP-geometry as a medium for constructing field theories, the above three principles are to be considered. The curves (paths) of
  the geometry have been derived \cite{W98} and used to describe the motion of spinning particles \cite{WMK2000, WMK2001}. Geometric objects have been
  used to represent physical quantities via the constructed field theories \cite{WKamal2012, WYElHanafy2014}.

The present work is done as a step to complete the use of the  geometrization philosophy in PAP-geometry. The general form of the differential identity
characterizing this geometry is obtained in section 3. This will help in attributing physical properties to geometric objects, especially conservation.
This will be discussed in the next point.

  \item [\textbf{(4)}] The results obtained from the two theorems given in section 4 are of special importance. It is clear from \eqref{corolla3}, \eqref{4.4}
  and \eqref{4.5} that any symmetric tensor defined by \eqref{E-definition} is subject to the differential identity of the type \eqref{4.4}.
  This result is independent of the values of the parameter $b$.

Now, for any symmetric tensor defined by \eqref{E-definition} in PAP-geometry, we have
\begin{equation}\label{5.1}
\s{E} {^{\mu\nu}}{_{;\mu}}\equiv 0 \,.
\end{equation}
Due to the structure of the parameterized connection \eqref{p-canonical}, the tensor $\s{E} {^{\mu\nu}}$ can always be written in the form
\begin{equation*}\label{5.2}
  \s{E} {^{\mu\nu}}= G ^{\mu\nu} +\s{T} {^{\mu\nu}}
\end{equation*}
where $G ^{\mu\nu}$ is Einstein tensor defined in terms of \eqref{Christoffel} as usual and $\s{T} {^{\mu\nu}}$ is a second order symmetric tensor
which vanishes when $b=0$. Consequently, the identity \eqref{5.1} can be written as
\begin{equation*}\label{5.3}
 G ^{\mu\nu}{_{~;\mu}}+
 \s{T} {^{\mu\nu}}{_{;\mu}}\equiv 0 \,.
\end{equation*}
As the first term vanishes, due to Bianchi identity, then we get the identity
\begin{equation}\label{5.4}
   \s{T} {^{\mu\nu}}{_{;\mu}}\equiv 0 \,.
\end{equation}
This implies that the physical entity represented by $ \s{T} {^{\mu\nu}}$ is conserved in any field theory constructed in PAP-geometry.
The identity \eqref{5.4} is thus an identity of PAP-geometry, independent of any such field theory. The tensor $ \s{T} {^{\mu\nu}}$ is usually
used as a geometric representation of a material-energy distribution in any field theory constructed in PAP-geometry \cite{WKamal2012, WYElHanafy2014}.

  \item [\textbf{(5)}] Finally, the following properties are guaranteed in any field theory constructed in the context of PAP-geometry:

\textbf{a)} The theory of GR, with all its consequences, can be obtained upon taking $ b=0$.

\textbf{b)} As the parameterized connection \eqref{p-canonical} is metric, the motion along curves of PAP-geometry \cite{W98} preserves the gravitational potential.

\textbf{c)} One can always define a geometric material energy tensor in terms of the BB of the geometry.

\textbf{d)} Conservation is not violated in any of such theories.
\end{itemize}


\end{document}